\documentclass[a4paper,UKenglish,cleveref,hyperref,autoref]{lipics-v2021}

\title{Tight Inapproximability of Max Independent Set in Triangle-Free Graphs}

\titlerunning{Tight Inapproximability of Max Independent Set in Triangle-Free Graphs}

\author{\'{E}douard Bonnet}{Univ Lyon, CNRS, ENS de Lyon, Université Claude Bernard Lyon 1, LIP UMR5668, France \and \url{http://perso.ens-lyon.fr/edouard.bonnet/}}{edouard.bonnet@ens-lyon.fr}{https://orcid.org/0000-0002-1653-5822}{}

\authorrunning{\'E. Bonnet}

\Copyright{Édouard Bonnet}

\category{}

\relatedversion{}

\supplement{}

\funding{}
 
\acknowledgements{}

\nolinenumbers 

\hideLIPIcs  

\EventEditors{John Q. Open and Joan R. Access}
\EventNoEds{2}
\EventLongTitle{42nd Conference on Very Important Topics (CVIT 2016)}
\EventShortTitle{CVIT 2016}
\EventAcronym{CVIT}
\EventYear{2016}
\EventDate{December 24--27, 2016}
\EventLocation{Little Whinging, United Kingdom}
\EventLogo{}
\SeriesVolume{42}
\ArticleNo{23}

\usepackage[utf8]{inputenc}  

\usepackage[T1]{fontenc}
\usepackage{lmodern}

\usepackage[colorinlistoftodos,bordercolor=orange,backgroundcolor=orange!20,linecolor=orange,textsize=normalsize]{todonotes}

\usepackage{amsmath}  
\usepackage{amssymb}     
\usepackage{bbm}
\usepackage{accents}
\usepackage{complexity}
\usepackage{booktabs}
\usepackage{fixmath}
\usepackage{paralist}
\usepackage{bm}

\makeatletter
\newtheorem*{rep@theorem}{\rep@title}
\newcommand{\newreptheorem}[2]{%
\newenvironment{rep#1}[1]{%
 \def\rep@title{#2 \ref{##1}}%
 \begin{rep@theorem}}%
 {\end{rep@theorem}}}
\makeatother

\newreptheorem{theorem}{Theorem}
\newreptheorem{lemma}{Lemma}

\crefname{observation}{Observation}{Observations}

\usepackage{xspace}
\usepackage{tikz}

\usepackage[ruled,vlined,linesnumbered]{algorithm2e}

\usetikzlibrary{fit}
\usetikzlibrary{arrows}
\usetikzlibrary{patterns}
\usetikzlibrary{calc}
\usetikzlibrary{shapes}
\usetikzlibrary{positioning}
\usetikzlibrary{math}
\usetikzlibrary{shapes.symbols}
\usetikzlibrary{decorations.pathreplacing}
\tikzset{draw half paths/.style 2 args={%
  decoration={show path construction,
    lineto code={
      \draw [#1] (\tikzinputsegmentfirst) -- 
         ($(\tikzinputsegmentfirst)!0.5!(\tikzinputsegmentlast)$);
      \draw [#2] ($(\tikzinputsegmentfirst)!0.5!(\tikzinputsegmentlast)$)
        -- (\tikzinputsegmentlast);
    }
  }, decorate
}}

\usepackage[scr=boondox,scrscaled=1.05]{mathalfa}

\renewcommand{\geq}{\geqslant}
\renewcommand{\leq}{\leqslant}

\usepackage{fontawesome5}

\newcommand{\mis}{\textsc{Max Independent Set}\xspace}

\newcommand{\smis}{\textsc{MIS}\xspace}

\newcommand{\vbl}{\operatorname{vbl}}

\theoremstyle{definition}

\newcommand{\Oh}{\mathcal O}

\newcommand\abs[1]{\lvert #1\rvert}

\begin{document}

\maketitle

\begin{abstract}
  For every $\varepsilon > 0$, it is NP-hard to $n^{1-\varepsilon}$-approximate \textsc{Max Independent Set} in $n$-vertex graphs [H{\aa}stad '96, Zuckerman '07].
  In triangle-free graphs, a~simple argument gives a~polynomial-time $n^{\frac{1}{2}}$-approximation algorithm, whereas, for every $\varepsilon > 0$, an $n^{\frac{1}{4}-\varepsilon}$-approximation algorithm would imply that $\mathrm{NP}\subseteq\mathrm{BPP}$ [Bonnet, Thomassé, Tran, Watrigant; ESA '20].

  In this note, we close this gap by proving the corresponding hardness against $n^{\frac{1}{2}-\varepsilon}$-approximation algorithms.
  The reduction is very simple and uses the Moser--Tardos resampling algorithm to make the constructed graphs triangle-free.
  The soundness uses a~result of Haeupler, Saha, and Srinivasan building on the proof of Moser and Tardos, to upper-bound the probability that a~fixed relatively large subset is an independent set after the Moser--Tardos algorithm terminates. 

  We generalize this scheme and show that, for any nonempty finite family $\mathcal F$ of graphs, each containing at~least one cycle, for any $\varepsilon > 0$, an $n^{\mu(\mathcal F)-\varepsilon}$-approximation algorithm for \textsc{Max Independent Set} in graphs excluding every member of~$\mathcal F$ as a~subgraph implies that $\mathrm{NP}\subseteq\mathrm{BPP}$, where
  \[
  \mu(\mathcal F):=1 - \max_{H \in\mathcal F}\ \min_{\substack{U\subseteq V(H)\\ H[U]\text{ contains a~cycle}}}\frac{\abs{U}-2}{\abs{E(H[U])}-1}.
\]
\end{abstract}

\section{Introduction}\label{sec:intro}

Given a~simple graph~$G$, the \mis problem (or \smis for short) asks for a~largest \emph{independent set} in~$G$, that is, a~subset of vertices that are pairwise nonadjacent.
This problem is notoriously inapproximable: For every $\varepsilon > 0$, it is NP-hard to \mbox{$n^{1-\varepsilon}$}-approximate in $n$-vertex graphs~\cite{Hastad96,Zuckerman07}.

It is conjectured in~\cite{Bonnet20} that the class of all graphs is the only \emph{hereditary} (i.e., closed under taking induced subgraphs) graph class on which \textsc{Max Independent Set} is that inapproximable.
More formally, the conjecture is that for every graph $H$, there exists a~constant $\delta_H \in [0,1)$ such that \smis admits a~(randomized) polynomial-time $\Oh(n^{\delta_H})$-approximation algorithm in $H$-free graphs (i.e., the graphs excluding $H$ as an induced subgraph).
This is shown to be a~weakening of an algorithmic version of the celebrated Erd\H{o}s--Hajnal conjecture.

We are here interested in tight approximability of \smis in hereditary classes, starting with $H$-free classes and $\mathcal F$-free classes (i.e., excluding every graph in~$\mathcal F$ as an induced subgraph) for a~finite family $\mathcal F$.
More specifically, we wish to match an $\Oh(n^{\delta_H})$-approximation algorithm with the conditional lower bound that, for any constant $\varepsilon > 0$, an $\Oh(n^{\delta_H-\varepsilon})$-approximation algorithm contradicts some standard complexity-theoretic assumption.
Prior to our work, no such result was known for any graph~$H$.

For triangle-free graphs (that is, when $H$ is a~triangle), there is a~straightforward $\Oh(n^{\frac{1}{2}})$-approximation algorithm.
If the input graph has a~vertex of degree at~least $n^{\frac{1}{2}}$, its neighborhood, which by triangle-freeness is an independent set, yields the desired approximate solution.
Otherwise the graph has maximum degree at~most~$n^{\frac{1}{2}}$, and every maximal independent set is large enough.
On the complexity side, it was shown that, for every $\varepsilon>0$, a~polynomial-time $n^{\frac{1}{4}-\varepsilon}$-approximation algorithm would imply that $\mathrm{NP} \subseteq \mathrm{BPP}$~\cite{Bonnet20}.
We close this gap and provide the first example of polynomial-factor tightness for \smis in $H$-free graphs.

\begin{theorem}\label{thm:main}
  For every constant $\varepsilon>0$, there is no polynomial-time $n^{\frac{1}{2}-\varepsilon}$-approximation algorithm for \mis on $n$-vertex triangle-free graphs unless $\mathrm{NP} \subseteq\mathrm{BPP}$.
\end{theorem}

The reduction behind~\cref{thm:main} starts with a~hard general \smis-instance~$H$, blows every vertex into an independent set of size~$n$, and independently samples (i.e., includes) every edge of the intermediate blow-up $F$ with probability $\Theta(1/n)$.
While the current graph contains a triangle, the algorithm resamples the three variables corresponding to its edges, thereby following the Moser--Tardos algorithm.
In expectation, this process terminates after $n^{\Oh(1)}$ steps, by the work of Moser and Tardos~\cite[Theorem~1.2]{Moser10} (see~\cref{thm:moser-tardos}).
The obtained graph $G$ has $n^2$ vertices and is, by design, triangle-free.

It is immediate that $\alpha(G) \geqslant \alpha(H) \cdot n$.
On the other hand, we show that, with high probability, $\alpha(G) = \Oh(\alpha(H) \cdot n \log n)$. 
This uses a~follow-up work by Haeupler, Saha, and Srinivasan~\cite[Theorem~2.2]{Haeupler11} (see~\cref{thm:hss}), building on~\cite{Moser10}, in order to upper-bound the probability, for any fixed set $S$ of $\Omega(\alpha(H) \cdot n \log n)$ vertices of~$G$, that $S$ is an independent set in~$G$.
We then conclude with a~union bound over all sets of $\Theta(\alpha(H) \cdot n \log n)$ vertices of~$G$.
Finally, the inapproximability of \smis by H{\aa}stad~\cite{Hastad96} and Zuckerman~\cite{Zuckerman07} (see~\cref{thm:zuckerman}) yields~\cref{thm:main}.

\medskip

We believe that the elementary proof of \cref{thm:main} is an accessible application of the Moser--Tardos algorithm/theorem, which can be taught within various related courses (approximation algorithms, randomness in computer science, graph independent sets, etc.).    
We decided that the current paper would primarily present this particular and simpler case, but the proof scheme naturally generalizes to $\mathcal F$-free graphs.
Besides, \cref{thm:main} remains the only example of polynomial-factor tight approximability for \smis in graphs excluding finitely many induced subgraphs.

For every nonempty finite family $\mathcal F$ of graphs, each with a~cycle, we set
\[
  \mu(\mathcal F):=1 - \max_{H \in\mathcal F}\ \min_{\substack{U\subseteq V(H)\\ H[U]\text{ contains a~cycle}}}\frac{\abs{U}-2}{\abs{E(H[U])}-1}.
\]
The proof technique sketched after \cref{thm:main} works for $\mathcal F$-free graphs, and even for \mbox{$\mathcal F$}-subgraph-free graphs, i.e., those excluding every graph of~$\mathcal F$ as a~mere subgraph.
We obtain the following theorem.

\begin{theorem}\label{thm:generalization}
  Let $\mathcal F$ be a nonempty finite family of graphs, each containing a~cycle.
  For every constant $\varepsilon>0$, there is no polynomial-time $n^{\mu(\mathcal F)-\varepsilon}$-approximation algorithm for \mis on $n$-vertex $\mathcal F$-subgraph-free graphs unless $\mathrm{NP}\subseteq\mathrm{BPP}$.
\end{theorem}

This theorem has some interesting consequences.
For $\mathcal F=\{K_t\}$, where $K_t$ is the $t$-vertex clique and $t \geqslant 3$,
\[\mu(\mathcal F) = 1 - \frac{t-2}{\binom{t}{2}-1} =  1 - \frac{t-2}{\frac{t(t-1)}{2}-1} =  1 - \frac{t-2}{\frac{(t-2)(t+1)}{2}}=1-\frac{2}{t+1}=\frac{t-1}{t+1}.\]
On the positive side, \smis admits an~$n^{\frac{t-2}{t-1}}$-approximation algorithm in $K_t$-free graphs~\cite{Bonnet20}.
For every integer $t \geqslant 4$, there is still a~gap between the hardness exponent $\frac{t-1}{t+1} - \varepsilon$ and the achieved exponent $\frac{t-2}{t-1}$; note that for $t=3$, the triangle-free case, it happens that $\frac{t-1}{t+1} = \frac{t-2}{t-1}$.

For $\mathcal F=\{K_{s,t}\}$, where $K_{s,t}$ is the biclique with sides of size $s$ and $t$, with $2 \leqslant s \leqslant t$,
\[\mu(\mathcal F) = 1 - \frac{s+t-2}{st-1} = \frac{st-s-t+1}{st-1} = \frac{(s-1)(t-1)}{st-1} = \frac{s-1}{s} \cdot \frac{t-1}{t-\frac{1}{s}}.\]
This nearly matches (in fact, asymptotically matches when $t$ goes to infinity) the existing $\Oh_{s,t}(n^{\frac{s-1}{s}})$-approximation algorithms \cite{Dvorak23,Bonnet20}.

For $\mathcal F=\{C_3, C_4, \ldots, C_{\gamma-1}\}$ where $\gamma \geqslant 4$ and $C_\ell$ is the $\ell$-vertex cycle, the $\mathcal F$-free graphs are precisely the graphs of girth at~least~$\gamma$, and
\[\mu(\mathcal F) = 1 - \frac{\gamma-3}{\gamma-2} = \frac{1}{\gamma-2}.\] 
Thus, \cref{thm:generalization} improves the hardness exponent $\frac{1}{2\gamma-2}-\varepsilon$ shown in~\cite{Bonnet20} for graphs of girth at~least~$\gamma$, but does not settle the approximability of \smis on these classes.
Indeed, the exponent of the best approximation factor is $\frac{2}{\gamma-1}$ when $\gamma$ is odd, and $\frac{2}{\gamma}$ when $\gamma$ is even~\cite{Monien85,Murphy92}. 

We leave the closing of these gaps as open problems.

\section{Preliminaries}\label{sec:prelim}

We denote by $V(G)$ and $E(G)$ the vertex and edge sets, respectively, of a~graph~$G$.
For every $X \subseteq V(G)$, $G[X]$ is the subgraph of~$G$ induced by~$X$, i.e., the graph obtained from~$G$ after removing all the vertices not in~$X$.
We denote by $\alpha(G)$ the \emph{independence number} of~$G$, that is, the largest cardinality of an \emph{independent set} of~$G$.

\medskip

Following the notation in~\cite{Moser10,Haeupler11}, we let $\mathcal P$ be a finite family of mutually independent random variables, and $\mathcal{A}$ be a~finite family of events, informally referred to as ``bad'' events, determined by these variables.
For any event $B$ (not necessarily in~$\mathcal A$) determined by $\mathcal P$, let $\vbl(B)$ be the smallest subset of~$\mathcal P$ determining $B$, and define
\[ \Gamma_{\mathcal A}(B) := \bigl\{A\in\mathcal A\setminus\{B\} \mid \vbl(A)\cap\vbl(B)\neq\varnothing \bigr\}. \]
When $\mathcal A$ is clear from the context, we simply write $\Gamma(B)$.

The Moser--Tardos resampling algorithm proceeds as follows.
Initially sample all variables in $\mathcal P$ independently according to their respective distributions.
While some event $A \in \mathcal A$ holds, resample all variables in $\vbl(A)$.

\begin{theorem}[{\cite[Theorem~1.2]{Moser10}}]\label{thm:moser-tardos}
  Let $\lambda \colon\mathcal A \to (0,1)$ be a~map such that for every $A \in \mathcal A$,
  \begin{equation}\label{eq:mt}
    \Pr[A] \leq \lambda(A) \prod_{B \in\Gamma(A)} \bigl (1-\lambda(B)\bigr).
  \end{equation}
  Then the expected number of resampling steps of the Moser--Tardos algorithm is at most \[\sum_{A\in\mathcal A}\frac{\lambda(A)}{1-\lambda(A)}.\]
\end{theorem}
We will also need to bound the probability that an event outside~$\mathcal A$ holds at some point during the execution of the Moser--Tardos algorithm.
Haeupler, Saha, and Srinivasan observe that the following can be derived from the proof of Moser and Tardos.

\begin{theorem}[{\cite[Theorem~2.2]{Haeupler11}}]\label{thm:hss}
  Assume $\mathcal A, \lambda$ satisfy the hypotheses of Theorem~\ref{thm:moser-tardos}.
  Let $B$ be any event determined by the variables in $\mathcal P$.
  The probability that $B$ is true at least once during the execution of the Moser--Tardos algorithm is at most
  \begin{equation}\label{eq:hss}
    \Pr[B] \prod_{A\in\Gamma(B)} \bigl(1-\lambda(A)\bigr)^{-1}.
  \end{equation}
\end{theorem}
A fortiori, the probability that $B$ holds when the algorithm terminates is at most the quantity in~\eqref{eq:hss}.
Finally, we will rely on the inapproximability of \textsc{Max Independent Set} in general graphs.

\begin{theorem}[promise version of {\cite[Theorem~1.1]{Zuckerman07} on complements}]\label{thm:zuckerman}
  For every constant $\delta\in(0,1/2)$, it is NP-hard to distinguish, given an $n$-vertex graph $H$, between the following two cases:
  $\alpha(H)\leq n^\delta~\text{and}~\alpha(H)\geq n^{1-\delta}$.
\end{theorem}

\section{Construction of the triangle-free graph}\label{sec:construction}

The reduction is very simple.
Given an $n$-vertex graph $H$, we build an $n^2$-vertex triangle-free graph $G$ as follows.
First replace every vertex $v \in V(H)$ by an independent set $I_v$ of size $n \geqslant 1$.
This yields the intermediate graph $F$ where $xy \in E(F)$ if and only if $x \in I_u$ and $y \in I_v$ for some $uv \in E(H)$.

For every $e \in E(F)$, we introduce an independent Bernoulli variable
\[P_e \sim \text{Bernoulli}(p)~~\text{with}~~p := \frac{1}{10n},\]
i.e., $P_e=1$ with probability~$p$, and $P_e = 0$ otherwise.
The family \[\mathcal P := \{P_e \mid e \in E(F)\}\] is our set of independent random variables (as used in~\cref{sec:prelim}).

For every triangle $xyz$ in~$F$, let $A_{xyz}$ be the ``bad'' event that $P_{xy} = P_{yz} = P_{xz} = 1$.
Set \[\mathcal A := \{A_{xyz} \mid xy, yz, xz \in E(F)\}.\]
Let $F\langle \mathcal P \rangle$ be the graph with vertex set $V(F)$ and edge set $\{e \in E(F) \mid P_e = 1\}$.

We run the Moser--Tardos algorithm on $\mathcal P, \mathcal A$.
We set $G := F\langle \mathcal P_{\mathrm{out}} \rangle$, where $\mathcal P_{\mathrm{out}}$ is the assignment of the variables in~$\mathcal P$ when the algorithm terminates; see~\cref{fig:construction}.
By design, we have:

\begin{observation}\label{obs:triangle-free}
  $G$ is triangle-free.
\end{observation}

\begin{figure}[!ht]
  \centering
  \begin{tikzpicture}[
    vertex/.style={circle,fill=black,inner sep=1.4pt},
    cluster/.style={draw=black!45,rounded corners=2pt,inner sep=5pt},
    edge/.style={draw=black,semithick},
    map/.style={->,semithick}
  ]
    \coordinate (Hu) at (0,1);
    \coordinate (Hv) at (-.65,-.35);
    \coordinate (Hw) at (.65,-.35);
    \foreach \a/\b in {u/v,v/w,w/u} \draw[edge] (H\a) -- (H\b);
    \node[vertex,label=above:$u$] at (Hu) {};
    \node[vertex,label=below left:$v$] at (Hv) {};
    \node[vertex,label=below right:$w$] at (Hw) {};

    \foreach \prefix/\shift in {F/3.7,G/7.6} {
      \foreach \set/\x/\y in {u/0/1,v/-.75/-.55,w/.75/-.55} {
        \foreach \i/\offset in {1/-.26,2/0,3/.26} \coordinate (\prefix\set\i) at ({\shift+\x+\offset},\y);
      }
    }
    \foreach \a/\b in {u/v,u/w} \foreach \i in {1,2,3} \foreach \j in {1,2,3} \draw[edge] (F\a\i) -- (F\b\j);
    \foreach \i in {1,2,3} \foreach \j in {1,2,3} \draw[edge] (Fv\i) to[bend left=40] (Fw\j);
    \foreach \i in {1,2,3} {
      \draw[edge] (Gu\i) -- (Gv\i);
      \draw[edge] (Gv\i) to[bend left=40] (Gw\i);
    }
    \foreach \i/\j in {1/2,2/3,3/1} \draw[edge] (Gw\i) -- (Gu\j);
    \foreach \prefix in {F,G} \foreach \set in {u,v,w} \node[cluster,fit=(\prefix\set1)(\prefix\set3)] {};
    \foreach \prefix in {F,G} \foreach \set in {u,v,w} \foreach \i in {1,2,3} \node[vertex] at (\prefix\set\i) {};
    \foreach \shift in {3.7,7.6} {
      \node at (\shift,1.42) {$I_u$};
      \node at ({\shift-.75},-.95) {$I_v$};
      \node at ({\shift+.75},-.95) {$I_w$};
    }

    \draw[map] (.95,.25) -- node[above] {$\times n$} (2.45,.25);
    \draw[map] (4.95,.25) -- node[above] {$\mathrm{MT}, p$} (6.35,.25);
    \node at (0,-1.35) {$H$};
    \node at (3.7,-1.35) {$F$};
    \node at (7.6,-1.35) {$G$};
  \end{tikzpicture}
  \caption{Construction of the triangle-free $G$ from a~graph~$H$; here $H$ is a~triangle.}
  \label{fig:construction}
\end{figure}
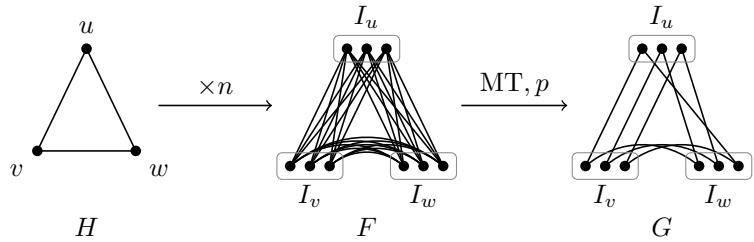

This randomized reduction does indeed terminate and take expected polynomial time. 

\begin{lemma}\label{lem:expected-building-time}
  In expectation, it takes $\Oh(n^9)$ time to build~$G$.
\end{lemma}

\begin{proof}
  One can build~$F$ from~$H$ in $\Oh(n^4)$ time.
  Set $\lambda:=2p^3 \in (0,1)$ and, by a slight abuse of notation, also use $\lambda$ for the constant map $A\mapsto2p^3$ on $\mathcal A$.
  Let us check that the pair $\mathcal A, \lambda$ satisfies~\cref{eq:mt}, so that we can apply~\cref{thm:moser-tardos}.

  For every $A_{xyz} \in \mathcal A$, we have $\abs{\Gamma(A_{xyz})} < d := 3n^2$.
  Indeed, there are fewer than $n^2$ triangles of~$F$ containing the edge $xy$ (and likewise for $yz$ and $xz$).
  By Bernoulli's inequality,
  \[(1-\lambda)^d \geqslant 1-d \lambda = 1 - 3n^2 \cdot 2p^3 = 1 - \frac{3}{500 n} > \frac{1}{2}.\]
  Therefore, for every $A \in \mathcal A$, 
  \[\lambda(1-\lambda)^d > \frac{\lambda}{2} = p^3 = \Pr[A],\]
  and~\cref{eq:mt} holds.
  By~\cref{thm:moser-tardos}, if $R$ is the number of resampling steps, then
  \[\mathbb E[R]\leqslant s:=\sum_{A\in\mathcal A}\frac{\lambda(A)}{1-\lambda(A)}=\abs{\mathcal A}\cdot\frac{\lambda}{1-\lambda}.\]
  We have $\abs{\mathcal A}\leqslant n^6$ and $(1-\lambda)^{-1}<2$, so
  \[s<2\abs{\mathcal A}\lambda=4\abs{\mathcal A}p^3\leqslant\frac{n^3}{250},\]
  and hence the expected number of resamplings is $\Oh(n^3)$; since $F$ has fewer than $n^6$ triangles and each search for a violated event can test them all in $\Oh(n^6)$ time, the total expected running time is $\Oh(n^9)$.
\end{proof}

By Markov's inequality, we get the following corollary.

\begin{corollary}\label{cor:building-time}
  $G$ is built in time $\Oh(n^{10})$ with high probability.
\end{corollary}

\section{Completeness and soundness}

If $S \subseteq V(H)$ is an independent set of~$H$, then $\bigcup_{v \in S} I_v$ is an independent set of $F$, hence of~$G$.
In particular, since $\abs{\bigcup_{v \in S} I_v} = \abs{S} n$, we have
\begin{observation}\label{obs:alpha-lb}
  $\alpha(G) \geqslant \alpha(H) \cdot n$.
\end{observation}

We now prove the following lemma.
\begin{lemma}\label{lem:alpha-ub}
  $\alpha(G) \leqslant \lceil 320\,\alpha(H) \cdot n \ln n \rceil$ with high probability.
\end{lemma}

We fix a nonempty set $S\subseteq V(F)$.
Our goal is to upper-bound the probability that $S$ is an independent set in the output graph~$G$ produced when the Moser--Tardos algorithm terminates.
We will then use a~union bound to prove~\cref{lem:alpha-ub}.

We set $s_v := \abs{S \cap I_v}$ for every $v \in V(H)$.
Therefore, \[s := \abs{S} = \sum\limits_{v \in V(H)} s_v.\]
Let $m(S)$ denote the number of edges in $F[S]$.
Thus, \[m(S) = \sum\limits_{uv \in E(H)} s_u s_v.\]
We finally set $z_v := s_v/s$.

We have $\sum_{v \in V(H)} z_v = \frac{1}{s} \sum_{v \in V(H)} s_v = 1$ and $z_v \geqslant 0$ for every $v \in V(H)$.
Therefore, the Motzkin--Straus theorem applied to $\overline H, (z_v)_{v \in V(\overline H)}$ yields the following inequality.

\begin{lemma}[{\cite[Theorem~1]{Motzkin65} applied to the complement~$\overline H$}]\label{lem:compl}
  $\sum\limits_{uv \in E(\overline H)} z_u z_v \leqslant \frac{1}{2} \left(1-\frac{1}{\alpha(H)}\right)$.
\end{lemma}

We use this lemma to lower-bound the number of edges of $F[S]$ when $S$ is relatively large.

\begin{lemma}\label{lem:many-edges}
  If $s \geqslant 2n \cdot \alpha(H)$, then $m(S) \geqslant \frac{s^2}{4 \alpha(H)}$.
\end{lemma}

\begin{proof}
  Since $\sum_{v \in V(H)} z_v = 1$, we have
  \[1=\left(\sum\limits_{v\in V(H)}z_v\right)^2=\left(\sum\limits_{v\in V(H)}z_v^2\right)+2\sum\limits_{\{u,v\}\in\binom{V(H)}{2}}z_uz_v.\]
  Together with~\cref{lem:compl}, this identity yields
  \[\sum\limits_{uv\in E(H)}z_uz_v
  =\sum\limits_{\{u,v\}\in\binom{V(H)}{2}}z_uz_v-\sum\limits_{uv\in E(\overline H)}z_uz_v
  \geqslant\frac{1}{2}\left(1-\sum\limits_{v\in V(H)}z_v^2\right)-\frac{1}{2}\left(1-\frac{1}{\alpha(H)}\right).\]
  Multiplying by $s^2$, we get
  \[m(S)=\sum\limits_{uv\in E(H)}s_us_v\geqslant\frac{s^2}{2\alpha(H)}-\frac{1}{2}\sum\limits_{v\in V(H)}s_v^2.\]
  Since $s_v \leqslant n$ for every $v \in V(H)$,
  \[m(S) \geqslant  \frac{s^2}{2 \alpha(H)} - \frac{n}{2} \sum\limits_{v \in V(H)} s_v = \frac{s^2}{2 \alpha(H)} - \frac{ns}{2}.\]
  Thus, if $s \geqslant 2n \cdot \alpha(H)$, we have $\frac{ns}{2} \leqslant \frac{s^2}{4 \alpha(H)}$, and we conclude that 
  \[m(S) \geqslant \frac{s^2}{2 \alpha(H)} - \frac{s^2}{4 \alpha(H)} = \frac{s^2}{4 \alpha(H)}.\qedhere\]
\end{proof}
Let $\mathcal I_S$ be the event that $P_e=0$ for every $e\in E(F[S])$.
We have
\[\Pr[\mathcal I_S]=(1-p)^{m(S)}.\]
Let $\mathcal J_S$ be the event that $\mathcal I_S$ holds for the output assignment, equivalently, that $S$ is an independent set in~$G$.

\begin{lemma}\label{lem:JS}
  If $s \geqslant 2n \cdot \alpha(H)$, then $\Pr[\mathcal J_S] \leqslant \exp \left(-\frac{s^2}{80 n\,\alpha(H)}\right)$.
\end{lemma}
\begin{proof}
  Again, set $\lambda:=2p^3$ and, by a slight abuse of notation, also use $\lambda$ for the constant map $A\mapsto2p^3$ on $\mathcal A$.
  We already checked in the proof of~\cref{lem:expected-building-time} that $\mathcal A$ and $\lambda$ satisfy the hypotheses of~\cref{thm:moser-tardos}.
  Thus, applying~\cref{thm:hss} to $\mathcal I_S$ and evaluating this event on the output assignment gives
  \[\Pr[\mathcal J_S]\leqslant\Pr[\mathcal I_S]\cdot\prod\limits_{A\in\Gamma(\mathcal I_S)}(1-\lambda(A))^{-1}
  =\Pr[\mathcal I_S]\cdot(1-\lambda)^{-\abs{\Gamma(\mathcal I_S)}}.\]
  There are fewer than $m(S) n^2$ triangles of~$F$ with at~least one edge in $F[S]$.
  Therefore, $\abs{\Gamma(\mathcal I_S)} \leqslant m(S) n^2$.
  We thus obtain that
  \[\Pr[\mathcal J_S]\leqslant(1-p)^{m(S)}\cdot(1-\lambda)^{-m(S)n^2}.\]
  If $\Pr[\mathcal J_S]=0$, the claimed bound is immediate.
  Otherwise, taking natural logarithms gives
  \[\ln\Pr[\mathcal J_S]\leqslant m(S)\ln(1-p)-m(S)n^2\ln(1-\lambda).\]
  Since $p, \lambda \in (0,1/2)$, we have $\ln(1-p) \leqslant -p$ and $-\ln(1-\lambda) \leqslant 2\lambda$, hence
  \[\ln \Pr[\mathcal J_S] \leqslant -m(S)p + 2\lambda m(S) n^2 = -\frac{m(S)}{10n} + \frac{m(S)}{250n} \leqslant -\frac{m(S)}{20n}.\]

  If $s \geqslant 2n \cdot \alpha(H)$, \cref{lem:many-edges} implies that $m(S) \geqslant \frac{s^2}{4 \alpha(H)}$, and thus
   \[ \Pr[\mathcal J_S] \leqslant \exp \left(-\frac{s^2}{80n\,\alpha(H)}\right).\qedhere\]
\end{proof}

We can now prove the main result of this section.

\begin{proof}[Proof of~\cref{lem:alpha-ub}]
  We set the threshold $s_0 := \lceil 320\,\alpha(H)\,n \ln n \rceil$.
  Our goal is to show that \[\Pr[\alpha(G) \leqslant s_0]=1-o(1).\]
  This holds trivially if $s_0>n^2$; we therefore assume that $s_0\leqslant n^2$.

  Note that $\Pr[\alpha(G) \geqslant s_0]$ is equal to the probability that there is at~least~one independent set of size~$s_0$ in~$G$.
  For all sufficiently large~$n$, we have $s_0\geqslant 2n\,\alpha(H)$, so a~union bound and~\cref{lem:JS} give
  \[ \Pr[\alpha(G) \geqslant s_0] \leqslant \binom{n^2}{s_0} \exp \left(-\frac{{s_0}^2}{80n\,\alpha(H)}\right)
  \leqslant \exp \left(2s_0 \ln n -\frac{s_0}{80n\,\alpha(H)} \cdot s_0 \right) \]
  \[ \leqslant \exp \left(2s_0 \ln n - (4 \ln n) \cdot s_0 \right) = \exp \left(-2s_0 \ln n \right) = o(1),\]
  and we conclude.
\end{proof}

\section{Wrapping up}

We can now complete the proof of~\cref{thm:main}.

\begin{reptheorem}{thm:main}
  For every constant $\varepsilon>0$, there is no polynomial-time $N^{\frac{1}{2}-\varepsilon}$-approximation algorithm for \mis on $N$-vertex triangle-free graphs unless $\mathrm{NP}\subseteq\mathrm{BPP}$.
\end{reptheorem}

\begin{proof}
  The claim is trivial for $\varepsilon>\frac{1}{2}$, so we may assume that $\varepsilon\in(0,\frac{1}{2}]$.
  Set $\delta := \frac{\varepsilon}{2} \in (0,\frac{1}{2})$.
    
  By~\cref{thm:zuckerman}, it is NP-hard to distinguish $n$-vertex graphs~$H$ with $\alpha(H)\geqslant n^{1-\delta}$ from those with $\alpha(H)\leqslant n^\delta$.
  Given~$H$, run the construction of~\cref{sec:construction} for at most $n^{10}$ computation steps, where $N:=n^2$.
  If the cutoff is reached, let~$G$ be the edgeless $N$-vertex graph; otherwise, let~$G$ be the output of the construction.
  By~\cref{lem:expected-building-time} and Markov's inequality, the cutoff is reached with probability $\Oh(n^{-1})=o(1)$.

  Thus, $G$ is always triangle-free, is built in time $\Oh(n^{10})=\Oh(N^5)$, and satisfies~\cref{obs:alpha-lb}.
  Suppose, for contradiction, that there is a polynomial-time $N^{\frac{1}{2}-\varepsilon}$-approximation algorithm for \mis on $N$-vertex triangle-free graphs.
  Let~$S$ be the independent set returned by this algorithm on~$G$.
  A~union bound on the cutoff event and the failure event in~\cref{lem:alpha-ub} gives, except with probability $o(1)$, \[\abs{S} \leqslant \alpha(G)\leqslant\left\lceil 320\,\alpha(H)n\ln n\right\rceil.\]
  The approximation guarantee and~\cref{obs:alpha-lb} give
  \[\abs{S} \geqslant \frac{\alpha(G)}{N^{\frac{1}{2}-\varepsilon}}\geqslant\alpha(H)n \cdot N^{\varepsilon-\frac{1}{2}}=\alpha(H)n^{2\varepsilon}.\]
  For all sufficiently large~$n$, if $\alpha(H)\leqslant n^\delta$, then, except with probability $o(1)$, \[\abs{S}\leqslant\left\lceil 320n^{1+\delta}\ln n\right\rceil<n^{1+3\delta}.\]
  If $\alpha(H)\geqslant n^{1-\delta}$, then \[\abs{S}\geqslant n^{1-\delta+2\varepsilon}=n^{1+3\delta}.\]
  Hence, accepting exactly when $\abs{S}\geqslant n^{1+3\delta}$ distinguishes the two cases with probability $1-o(1)$.
  The finitely many smaller inputs can be handled by brute force.
  Since the distinction is NP-hard, this yields $\mathrm{NP}\subseteq\mathrm{BPP}$.
\end{proof}

The same result holds for randomized polynomial-time $N^{\frac{1}{2}-\varepsilon}$-approximation algorithms with success probability at least $\frac{2}{3}$, since $\Oh(\log n)$ independent repetitions, retaining the largest valid independent set returned, amplify the success probability to $1-o(1)$.

\section{Generalization to $\mathcal F$-free graphs}\label{sec:generalization}

We recall that 
\[
  \mu(\mathcal F):=1 - \max_{K\in\mathcal F}\ \min_{\substack{U\subseteq V(K)\\ K[U]\text{ contains a~cycle}}}\frac{\abs{U}-2}{\abs{E(K[U])}-1}.
\]
Note that if every graph in~$\mathcal F$ contains a~cycle, then $\mu(\mathcal F)\in(0,1)$.
The proof of the following theorem closely follows the triangle-free case.

\begin{reptheorem}{thm:generalization}
  Let $\mathcal F$ be a nonempty finite family of graphs, each containing a~cycle.
  For every constant $\varepsilon>0$, there is no polynomial-time $N^{\mu(\mathcal F)-\varepsilon}$-approximation algorithm for \mis on $N$-vertex $\mathcal F$-subgraph-free graphs unless $\mathrm{NP}\subseteq\mathrm{BPP}$.
\end{reptheorem}

\begin{proof}
  Set $\mu:=\mu(\mathcal F)$.
  For every $K\in\mathcal F$, choose $U_K\subseteq V(K)$ attaining the inner minimum in the definition of~$\mu$, and set
  \[
    J_K:=K[U_K],\qquad v_K:=\abs{V(J_K)},\qquad e_K:=\abs{E(J_K)}.
  \]
  Thus,
  \begin{equation}\label{eq:generalization-density}
    \frac{v_K-2}{e_K-1}\leqslant1-\mu.
  \end{equation}

  Let $H$ be an $n$-vertex graph.
  Put
  \[
    \ell:=\left\lceil n^{\frac{1-\mu}{\mu}}\right\rceil
    \qquad\text{and}\qquad
    N:=n\ell.
  \]
  As in~\cref{sec:construction}, replace every vertex $v\in V(H)$ by an independent set $I_v$ of size~$\ell$, and denote the resulting complete blow-up by~$F$.
  For every edge $e\in E(F)$, introduce an independent variable $P_e\sim\text{Bernoulli}(p)$, where
  \[
    p:=cN^{-(1-\mu)}\qquad\text{and}\qquad
    c := \frac{1}{\sqrt{16 e^2_*\abs{\mathcal F}}},
  \]
  with $e_* := \max\limits_{K \in \mathcal F} e_K$.
  Note that $p, c \in (0,1)$, and $c$ depends only on~$\mathcal F$.

\medskip

  For every $K\in\mathcal F$ and every copy~$Q$ of~$J_K$ in~$F$, not necessarily induced, let $A_Q$ be the event that $P_e=1$ for every $e\in E(Q)$.
  Let $\mathcal A$ be the family of all these bad events, and define
  \[
    \lambda(A_Q):=2p^{e_K} \in (0,1)
  \]
  whenever $Q$ is a~copy of~$J_K$.
  A~fixed edge of~$F$ belongs to at~most $2e_K N^{v_K-2}$ copies of~$J_K$.
  Hence, for every bad event~$A_Q$,
  \[
    \sum_{A\in\Gamma(A_Q)}\lambda(A)
    \leqslant e_* \sum_{K\in\mathcal F} 2e_K N^{v_K-2} \cdot 2p^{e_K} \leqslant 4e^2_* \sum_{K\in\mathcal F} N^{v_K-2}p^{e_K} \leqslant 4e^2_* \sum_{K\in\mathcal F} N^{v_K-2}p^{e_K-1}.
  \]
  Moreover,
  \begin{equation}\label{eq:generalization-weight}
    N^{v_K-2}p^{e_K-1} = c^{e_K-1}N^{v_K-2-(1-\mu)(e_K-1)} \leqslant c^{e_K-1} \leqslant c^2,
  \end{equation}
  where the first inequality follows from~\eqref{eq:generalization-density}, and the second follows from $c \leqslant 1$ and $e_K \geqslant 3$.
  Therefore,
  \[
    \sum_{A\in\Gamma(A_Q)}\lambda(A) \leqslant 4e^2_* |\mathcal F| c^2 \leqslant \frac{1}{2}.
  \]
  This justifies the second inequality in the following statement, while the first inequality holds by the Weierstrass product inequality:
  \[
    \lambda(A_Q)\prod_{A\in\Gamma(A_Q)}(1-\lambda(A))
    \geqslant 2p^{e_K}\left(1-\sum_{A\in\Gamma(A_Q)}\lambda(A)\right)
    \geqslant p^{e_K}
    =\Pr[A_Q].
  \]
  The hypotheses of~\cref{thm:moser-tardos} are therefore satisfied.
  Run the Moser--Tardos algorithm with $(P_e)_{e \in E(F)}, \mathcal A$ and let $G$ be the graph defined by its output assignment.
  (Technically, $p$ should be chosen as a~rational number with a~polynomial-length representation, sufficiently close to the displayed value that all the inequalities of this proof remain valid.)
  If $R$ denotes the number of resampling steps, then~\cref{thm:moser-tardos} and the bound $N^{v_K}$ on the number of copies of~$J_K$ give
  \[
    \mathbb E[R] \leqslant \sum_{A\in\mathcal A}\frac{\lambda(A)}{1-\lambda(A)}\leqslant 2 \cdot 2p^3 \sum_{K\in\mathcal F}N^{v_K}=N^{\Oh_{\mathcal F}(1)},
  \]
  where the second inequality holds since $\lambda(A) \leqslant 2p^3$ and $(1-\lambda(A))^{-1} \leqslant 2$.
  All bad events can be enumerated and tested in $N^{\Oh_{\mathcal F}(1)}$ time, so the expected running time is polynomial in~$N$ as well.
  Moreover, $G$ contains no $J_K$ as a~subgraph for any $K\in\mathcal F$.
  In particular, $G$ is $\mathcal F$-subgraph-free, since every copy of~$K$ would contain a~copy of $J_K=K[U_K]$.

  \medskip

  Completeness is unchanged from the triangle-free case:
  \begin{equation}\label{eq:generalization-completeness}
    \alpha(G)\geqslant \ell\alpha(H).
  \end{equation}
  We next establish the corresponding soundness bound.
  Fix a nonempty set $S\subseteq V(F)$, set $s:=\abs{S}$, and let $m(S):=\abs{E(F[S])}$.
  The same calculation as in~\cref{lem:many-edges} (with $\ell$ now replacing~$n$ as the size of each $I_v$) gives
  \[
    m(S)\geqslant\frac{s^2}{2\alpha(H)}-\frac{\ell s}{2}=\frac{s(s-\alpha(H)\ell)}{2\alpha(H)}.
  \]
  Thus, if $s \geqslant 2 \alpha(H) \ell$, then $s - \alpha(H)\ell \geqslant \frac{s}{2}$, so
  \begin{equation}\label{eq:generalization-many-edges}
    m(S)\geqslant\frac{s^2}{4\alpha(H)}.
  \end{equation}

  Let $\mathcal I_S$ be the event that $P_e=0$ for every $e\in E(F[S])$, and let $\mathcal J_S$ be the event that $S$ is an independent set in~$G$.
  For every fixed edge $e\in E(F)$, the copy count above gives
  \[
      \sum_{\substack{A_Q\in\mathcal A\\ \text{with}\,e \in E(Q)}} \lambda(A_Q)
      \leqslant \sum_{K\in\mathcal F} 2e_K N^{v_K-2} \cdot 2p^{e_K}
      \leqslant 4 e_* p \sum_{K\in\mathcal F} N^{v_K-2} p^{e_K-1}
      \leqslant p \cdot 4 e_* \abs{\mathcal F} c^2
      \leqslant \frac{p}{4},
  \]
  where the third inequality holds by~\eqref{eq:generalization-weight}.
  Every event in~$\Gamma(\mathcal I_S)$ depends on some variable $P_e$ with $e \in E(F[S])$, and therefore
  \begin{equation}\label{eq:generalization-IS-weight}
    \sum_{A\in\Gamma(\mathcal I_S)}\lambda(A)
    \leqslant\sum_{e\in E(F[S])}\ \sum_{\substack{A_Q\in\mathcal A\\ \text{with}\,e\in E(Q)}}\lambda(A_Q)
    \leqslant\frac{p\,m(S)}{4}.
  \end{equation}
  Moreover, $\lambda(A)<1/2$ for every $A \in \mathcal A$.
  We apply~\cref{thm:hss} and get 
   \[
      \Pr[\mathcal J_S]
      \leqslant \Pr[\mathcal I_S] \prod_{A\in\Gamma(\mathcal I_S)}(1-\lambda(A))^{-1}
      = (1-p)^{m(S)}\prod_{A\in\Gamma(\mathcal I_S)}(1-\lambda(A))^{-1}
  \]
  Since $1-p\leqslant\exp(-p)$, and $-\ln(1-x) \leqslant 2x$ for every $x \in (0,1/2)$, we have
  \[
      \Pr[\mathcal J_S]
      \leqslant \exp\left(-pm(S) - \sum_{A\in\Gamma(\mathcal I_S)}\ln (1-\lambda(A)) \right)
      \leqslant \exp\left(-pm(S) + 2 \sum_{A\in\Gamma(\mathcal I_S)} \lambda(A) \right)
  \]
  Therefore, by \eqref{eq:generalization-IS-weight},
  \[
      \Pr[\mathcal J_S]
   \leqslant \exp\left(-\frac{p\,m(S)}{2}\right).
  \]
  Together with~\eqref{eq:generalization-many-edges}, this shows that, whenever $s\geqslant2\alpha(H)\ell$,
  \begin{equation}\label{eq:generalization-independent-set}
    \Pr[\mathcal J_S]\leqslant\exp\left(-\frac{ps^2}{8\alpha(H)}\right).
  \end{equation}
  Set $s_0:=\left\lceil 32\,\alpha(H) p^{-1} \ln N \right\rceil$.
  If $s_0>N$, then $\alpha(G)<s_0$ holds trivially, so assume that $s_0\leqslant N$.

  Since $N=\Theta\left(n^{\frac{1}{\mu}}\right)$ and $p^{-1}=\Theta\left(N^{1-\mu}\right)=\Theta(\ell)$,
  we have $s_0 \geqslant 2\alpha(H)\ell$ for all sufficiently large~$n$.
  A~union bound using~\eqref{eq:generalization-independent-set} gives
  \[
      \Pr[\alpha(G)\geqslant s_0]
      \leqslant\binom{N}{s_0}\exp\left(-\frac{p{s_0}^2}{8\alpha(H)}\right)
      \leqslant\exp\left(s_0\ln N-4s_0\ln N\right)
      \leqslant\exp(-3s_0\ln N)=o(1).
  \]
  Therefore, with high probability,
  \begin{equation}\label{eq:generalization-soundness}
    \alpha(G)<s_0=\Oh(\alpha(H)\ell\ln N).
  \end{equation}

  \medskip

  It remains to turn this construction into a worst-case polynomial-time reduction and compare the gap.
  Fix a constant~$q$ such that the expected running time of the construction is $\Oh(N^q)$, and truncate it after $N^{q+1}$ computation steps.
  If the cutoff is reached, output the edgeless $N$-vertex graph.
  By Markov's inequality, the cutoff is reached with probability $\Oh(N^{-1})=o(1)$.

  Since every graph in~$\mathcal F$ contains a~cycle, the resulting graph is $\mathcal F$-subgraph-free whether or not the cutoff is reached, and~\eqref{eq:generalization-completeness} always holds.
  Set $\delta:=\min\{1/4,\varepsilon/(4\mu)\}\in(0,1/2)$.
  By~\cref{thm:zuckerman}, it is NP-hard to distinguish $n$-vertex graphs~$H$ with $\alpha(H)\geqslant n^{1-\delta}$ from those with $\alpha(H)\leqslant n^\delta$.
  Equations~\eqref{eq:generalization-completeness} and~\eqref{eq:generalization-soundness} yield the gap
  \[
    \Omega\left(\frac{n^{1-2\delta}}{\ln N}\right)
    =N^{\mu(1-2\delta)-o(1)}.
  \]
  Since $\mu(1-2\delta)\geqslant\mu-\varepsilon/2>\mu-\varepsilon$, this is larger than $N^{\mu-\varepsilon}$ for all sufficiently large~$N$.
  Such an approximation algorithm would therefore distinguish the two cases with probability $1-o(1)$, and would imply that $\mathrm{NP}\subseteq\mathrm{BPP}$.
\end{proof}

\subparagraph*{AI disclosure.}
\cref{thm:main} was proven essentially autonomously by GPT-5.6 Sol Pro.
The author then suggested the generalization \cref{thm:generalization} and the definition of~$\mu(\mathcal F)$.
The write-up is due to the author, who tried to make the proof of~\cref{thm:main} a~pleasant and effortless read.


\begin{thebibliography}{1}

\bibitem{Bonnet20}
{\'E}douard Bonnet, St{\'e}phan Thomass{\'e}, Xuan~Thang Tran, and R{\'e}mi
  Watrigant.
\newblock An algorithmic weakening of the {Erd{\H{o}}s--Hajnal} conjecture.
\newblock In {\em 28th Annual European Symposium on Algorithms (ESA 2020)},
  volume 173 of {\em Leibniz International Proceedings in Informatics
  (LIPIcs)}, pages 23:1--23:18. Schloss Dagstuhl--Leibniz-Zentrum f{\"u}r
  Informatik, 2020.
\newblock \href {https://doi.org/10.4230/LIPIcs.ESA.2020.23}
  {\path{doi:10.4230/LIPIcs.ESA.2020.23}}.

\bibitem{Dvorak23}
Pavel Dvo{\v{r}}{\'a}k, Andreas~Emil Feldmann, Ashutosh Rai, and Pawe{\l}
  Rz{\k{a}}{\.z}ewski.
\newblock Parameterized inapproximability of independent set in {$H$}-free
  graphs.
\newblock {\em Algorithmica}, 85(4):902--928, 2023.
\newblock \href {https://doi.org/10.1007/s00453-022-01052-5}
  {\path{doi:10.1007/s00453-022-01052-5}}.

\bibitem{Haeupler11}
Bernhard Haeupler, Barna Saha, and Aravind Srinivasan.
\newblock New constructive aspects of the {Lov{\'a}sz} local lemma.
\newblock {\em Journal of the ACM}, 58(6):28:1--28:28, 2011.
\newblock \href {https://doi.org/10.1145/2049697.2049702}
  {\path{doi:10.1145/2049697.2049702}}.

\bibitem{Hastad96}
Johan H{\aa}stad.
\newblock Clique is hard to approximate within $n^{1-\varepsilon}$.
\newblock In {\em 37th Annual Symposium on Foundations of Computer Science,
  {FOCS} 1996, Burlington, Vermont, USA, 14-16 October, 1996}, pages 627--636.
  {IEEE} Computer Society, 1996.
\newblock \href {https://doi.org/10.1109/SFCS.1996.548522}
  {\path{doi:10.1109/SFCS.1996.548522}}.

\bibitem{Monien85}
Burkhard Monien and Ewald Speckenmeyer.
\newblock Ramsey numbers and an approximation algorithm for the vertex cover
  problem.
\newblock {\em Acta Informatica}, 22(1):115--123, 1985.
\newblock \href {https://doi.org/10.1007/BF00290149}
  {\path{doi:10.1007/BF00290149}}.

\bibitem{Moser10}
Robin~A. Moser and G{\'a}bor Tardos.
\newblock A constructive proof of the general {Lov{\'a}sz} local lemma.
\newblock {\em Journal of the ACM}, 57(2):11:1--11:15, 2010.
\newblock \href {https://doi.org/10.1145/1667053.1667060}
  {\path{doi:10.1145/1667053.1667060}}.

\bibitem{Motzkin65}
Theodore~S. Motzkin and Ernst~G. Straus.
\newblock Maxima for graphs and a new proof of a theorem of {Tur{\'a}n}.
\newblock {\em Canadian Journal of Mathematics}, 17:533--540, 1965.
\newblock \href {https://doi.org/10.4153/CJM-1965-053-6}
  {\path{doi:10.4153/CJM-1965-053-6}}.

\bibitem{Murphy92}
Owen~J. Murphy.
\newblock Computing independent sets in graphs with large girth.
\newblock {\em Discrete Applied Mathematics}, 35(2):167--170, 1992.
\newblock \href {https://doi.org/10.1016/0166-218X(92)90041-8}
  {\path{doi:10.1016/0166-218X(92)90041-8}}.

\bibitem{Zuckerman07}
David Zuckerman.
\newblock Linear degree extractors and the inapproximability of {Max Clique}
  and {Chromatic Number}.
\newblock {\em Theory of Computing}, 3(6):103--128, 2007.
\newblock \href {https://doi.org/10.4086/toc.2007.v003a006}
  {\path{doi:10.4086/toc.2007.v003a006}}.

\end{thebibliography}
\end{document}